\documentclass[10pt,conference]{IEEEtran}
\usepackage{graphicx}
\usepackage{times}
\usepackage{amsmath}
\usepackage{setspace}
\usepackage{subfigure}
\usepackage{longtable}
\usepackage{multicol}
\usepackage{amssymb}
\usepackage{epsfig,latexsym,subfigure,cite}
\usepackage{color}

\newtheorem{theorem}{Theorem}
\newtheorem{corollary}{Corollary}
\newtheorem{lemma}{Lemma}

\newtheorem{remark}{Remark}

\providecommand{\Cov}{{\rm Cov}}

\providecommand{\xv}{\mathbf{x}} 
 \providecommand{\vv}{\mathbf{v}}
 \providecommand{\Uv}{\mathbf{U}}
\providecommand{\Vv}{\mathbf{V}} \providecommand{\Xv}{\mathbf{X}}

\providecommand{\hv}{\mathbf{h}} \providecommand{\gv}{\mathbf{g}}
 
\providecommand{\ev}{\mathbf{e}} \providecommand{\cv}{\mathbf{c}}

 \providecommand{\Yc}{{\mathcal Y}}
\providecommand{\Wc}{{\mathcal W}} 
 
 \providecommand{\Cc}{{\mathcal C}}
\providecommand{\Pc}{{\mathcal P}} \providecommand{\Rc}{{\mathcal R}}
 \providecommand{\Sc}{{\mathcal S}}

\providecommand{\Yt}{\tilde{Y}} \providecommand{\yt}{\tilde{y}}
\providecommand{\Zt}{\tilde{Z}} \providecommand{\zt}{\tilde{z}}

\providecommand{\tr}{{\rm tr}}

\IEEEoverridecommandlockouts

\title{Multi-Antenna Gaussian Broadcast Channels with Confidential Messages}
\author{
\authorblockN{Ruoheng Liu and H. Vincent Poor}
\authorblockA{Department of Electrical Engineering, Princeton University,
Princeton, NJ 08544 \\
Email: \{rliu,poor\}@princeton.edu}
\thanks{This research was supported by the National Science Foundation under Grants
ANI-03-38807, CNS-06-25637 and CCF-07-28208.}%
}

\begin{document}

\maketitle
\thispagestyle{empty}


\begin{abstract}

In wireless data networks, communication is particularly susceptible to
eavesdropping due to its broadcast nature. Security and privacy systems have
become critical for wireless providers and enterprise networks. This paper
considers the problem of secret communication over a Gaussian broadcast
channel, where a multi-antenna transmitter sends independent confidential
messages to two users with \emph{information-theoretic secrecy}. That is, each
user would like to obtain its own confidential message in a reliable and safe
manner. This communication model is referred to as the multi-antenna Gaussian
broadcast channel with confidential messages (MGBC-CM). Under this
communication scenario, a secret dirty-paper coding scheme and the
corresponding achievable secrecy rate region are first developed based on
Gaussian codebooks. Next, a computable Sato-type outer bound on the secrecy
capacity region is provided for the MGBC-CM. Furthermore, the Sato-type outer
bound proves to be consistent with the boundary of the secret dirty-paper
coding achievable rate region, and hence, the secrecy capacity region of the
MGBC-CM is established. Finally, a numerical example demonstrates that both
users can achieve positive rates simultaneously under the information-theoretic
secrecy requirement.
\end{abstract}


\section{Introduction}

The demand for efficient, reliable, and secret data communication over wireless
networks has become increasingly critical in recent years. Due to its broadcast
nature, wireless communication is particularly susceptible to eavesdropping.
The inherent nature of wireless networks exposes not only vulnerabilities that
a malicious user can exploit to severely compromise the network, but also
multiplies information confidentiality concerns with respect to in-network
terminals. Hence, security and privacy systems have become critical for
wireless providers and enterprise networks.

In this work, we consider multiple antenna secret broadcast in wireless
networks. This research is inspired by the seminal paper \cite{Wyner:BSTJ:75},
in which Wyner introduced the so-called {\it wiretap channel} and proposed an
information theoretic approach to secret communication schemes. Under the
assumption that the channel to the eavesdropper is a degraded version of that
to the desired receiver, Wyner characterized the capacity-secrecy tradeoff for
the discrete memoryless wiretap channel and showed that secret communication is
possible without sharing a secret key. Later, the result was extended by
Csisz{\'{a}r and K{\"{o}rner who determined the secrecy capacity for the
non-degraded {\it broadcast channel} (BC) with a single confidential message
intended for one of the users \cite{Csiszar:IT:78}.
\begin{figure}[t]
 \centerline{\includegraphics[width=0.95\linewidth,draft=false]{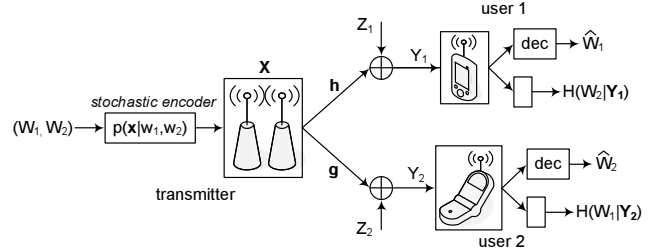}}
\caption{ Channel model of multiple-antenna Gaussian broadcast channel with
confidential messages } \label{fig:gbc}
\end{figure}

In more general wireless network scenarios, secret communication may involve
multiple users and multiple antennas. Consequently, a significant recent
research effort has been invested in the study of the information-theoretic
limits of secret communication in different wireless network environments
including multi-user communication with confidential messages
\cite{Csiszar:IT:04,Tekin:ITA:07,Liang06it,Liu07it,Lai:IT:06,Elza:CISS:07},
secret wireless communication on fading channels
\cite{Barros:ISIT:06,Liang06novit,Gopala:ISIT:07}, and the Gaussian
multiple-input single-output (MISO) and multiple-input multiple-output (MIMO)
wiretap channels
\cite{Li:CISS:07,Liu:WITS:07,Liu:IT:2007oct,Wornell-IT-2007,shafiee-IT-2007}.

These issues and results motivate us to study the multi-antenna Gaussian BC
with confidential messages (MGBC-CM), in which independent confidential
messages from a multi-antenna transmitter are to be communicated to two users.
The corresponding broadcast communication model is shown in Fig.~\ref{fig:gbc}.
Each user would like to obtain its own message reliably and confidentially.

To give insight into this problem, we first consider a single-antenna Gaussian
BC. Note that this channel is degraded~\cite{Cover}, which means that if a
message can be successfully decoded by the inferior user, then the superior
user is also ensured of decoding it. Hence, the secrecy rate of the inferior
user is zero and this problem is reduced to the scalar Gaussian wiretap channel
problem~\cite{Cheong:IT:78} whose secrecy capacity is now the maximum rate
achievable by the superior user. This analysis gives rise to the question: can
the transmitter, in fact, communicate with both users confidentially at nonzero
rate under some other conditions? Roughly speaking, the answer is in the
affirmative. In particular, the transmitter can communicate when equipped with
sufficiently separated multiple antennas.

We here have two goals motivated directly by questions arising in practice. The
first is to determine conditions under which both users can obtain their own
confidential messages 
in a reliable and safe manner. This is equivalent to evaluating the secrecy
capacity region for the MGBC-CM. The second is to show \emph{how} the
transmitter should broadcast confidentially, which is equivalent to designing
an achievable secret coding scheme. To this end, we first describe a secret
{\it dirty-paper coding} (DPC) scheme and derive the corresponding achievable
rate region based on Gaussian codebooks. The secret DPC is based on {\it
double-binning} \cite{Liu07it} which enables both joint encoding and preserving
confidentiality. Next, a computable Sato-type outer bound on the secrecy
capacity region is developed for the MGBC-CM. Furthermore, the Sato-type outer
bound proves to be consistent with the boundary of the secret dirty-paper
coding achievable rate region, and hence, the secrecy capacity region of the
MGBC-CM is established. Finally, a numerical example demonstrates that both
users can achieve positive rates simultaneously under the information-theoretic
secrecy requirement.


\section{System Model and Definitions} \label{sec:model}

\subsection{Channel Model}

We consider the communication of confidential messages to two users over a
Gaussian BC via $t\ge 2$ transmit-antennas. Each user is equipped with a single
receive-antenna. As shown in Fig.~\ref{fig:gbc}, the transmitter sends
independent confidential messages $W_1$ and $W_2$ in $n$ channel uses with
$nR_1$ and $nR_2$ bits, respectively. The message $W_1$ is destined for user 1
and eavesdropped upon by user 2, whereas the message $W_2$ is destined for user
2 and eavesdropped upon by user 1. This communication scenario is referred to
as the {\it multi-antenna Gaussian BC with confidential messages}. The Gaussian
BC is an additive noise channel and the received symbols at user 1 and user 2
can be represented as follows:
\begin{align}
y_{1,i}&=\hv^{H} \xv_i+z_{1,i} \notag\\
y_{2,i}&=\gv^{H} \xv_i+z_{2,i}, \qquad  i=1,\dots,n \label{eq:miso}
\end{align}
where $\xv_i \in \mathbb{C}^{t}$ is a complex input vector at time $i$,
$\{z_{1,i}\}$ and $\{z_{2,i}\}$ correspond to two independent, zero-mean,
unit-variance, complex Gaussian noise sequences, and $\hv, \gv\in
\mathbb{C}^{t}$ are fixed, complex channel attenuation vectors imposed on
user~1 and user~2, respectively. The channel input is constrained by
$\tr(K_{\Xv}) \le P$, where $K_{\Xv}$ is the channel input covariance matrix
and $P$ is the average total power limitation at the transmitter. We also
assume that both the transmitter and users are aware of the attenuation
vectors.

\subsection{Important Channel Parameters for the MGBC-CM}

For the MGBC-CM, we are interested in the following important parameters, which
are related to the generalized eigenvalue problem. Let $\lambda_1$ and $\ev_1$
denote the largest generalized eigenvalue and the corresponding normalized
eigenvector of the pencil $(I+P\hv \hv^{H}, I+P\gv \gv^{H})$ so that
$\ev_1^{H}\ev_1=1$ and
\begin{align}
(I+P\hv \hv^{H})\ev_1= \lambda_1(I+P\gv \gv^{H})\ev_1. \label{eq:eig-def1}
\end{align}
Similarly, we define $\lambda_2$ and $\ev_2$ as the largest generalized
eigenvalue and the corresponding normalized eigenvector of the pencil $(I+P\gv
\gv^{H}, I+P\hv \hv^{H})$ so that $\ev_2^{H}\ev_2=1$ and
\begin{align}
(I+P\gv \gv^{H})\ev_2= \lambda_2(I+P\hv \hv^{H})\ev_2. \label{eq:eig-def2}
\end{align}
A useful property of $\lambda_1$ and $\lambda_2$ is described as follows.
\begin{lemma} \label{lem:geig}
For any channel attenuation vector pair $\hv$ and $\gv$, the largest
generalized eigenvalues of the pencil $(I+P\hv \hv^{H}, I+P\gv \gv^{H})$ and
the pencil $(I+P\gv \gv^{H}, I+P\hv \hv^{H})$ satisfy $\lambda_1 \ge 1$ and
$\lambda_2 \ge 1.$ Moreover, if $\hv$ and $\gv$ are linearly independent, then
both $ \lambda_1$ and $\lambda_2$ are strictly greater than $1$.
\end{lemma}

\subsection{Definitions}

We now define the secret codebook, the probability of error, the secrecy level,
and the secrecy capacity region for the MGBC-CM as follows.

An $(2^{nR_1},2^{nR_2},n)$ {\it secret codebook} for the MGBC-CM consists of
the following:
\begin{enumerate}
  \item Two message sets: $\Wc_k=\{1,\ldots,2^{nR_k}\}$, for $k=1,2.$

  \item A stochastic encoding function specified by a
  conditional probability density $p(\xv^{n}|w_1,w_2)$,
  where $\xv^{n}=[\xv_1,\ldots,\xv_n]\in \mathbb{C}^{t\times n}$, $w_k\in\Wc_k$ for $k=1,2$, and
\begin{align*}
\int_{\xv^{n}}p(\xv^{n}|w_1,w_2)=1.
\end{align*}

  \item Decoding functions $\phi_1$ and $\phi_2$. The decoding function
  at user $k$ is a deterministic mapping $\phi_k  : \Yc_k^{n}\rightarrow
  \Wc_k.$

\end{enumerate}

At the receiver ends, the error performance and the secrecy level are evaluated
by the following performance measures.
\begin{enumerate}
  \item The reliability is measured by the maximum error probability
  $$P_e^{(n)}\triangleq \max \bigl\{P_{e,1}^{(n)}, P_{e,2}^{(n)}\bigr\}$$
  where $P_{e,k}^{(n)}$ is the error probability for user $k$.

  \item The secrecy levels with respect to confidential messages $W_1$ and $W_2$
are measured, respectively, at user~2 and user 1 with respect to the {\it
equivocation rates} $\frac{1}{n}H(W_2|Y_1^{n})$ and
$\frac{1}{n}H(W_1|Y_2^{n})$.
\end{enumerate}

A rate pair $(R_1,  R_2)$ is said to be achievable for the MGBC-CM if, for any
$\epsilon>0$, there exists an $(2^{nR_1}, 2^{nR_2}, n)$ code that satisfies
$P_e^{(n)}\le \epsilon$, and the information-theoretic secrecy
requirement
\begin{align}
nR_1-H(W_1|Y_2^{n})&\le n\epsilon ~\text{and} ~ nR_2-H(W_2|Y_1^{n})\le
n\epsilon. \label{eq:equiv}
\end{align}

The {\it secrecy capacity region} $\Cc^{\rm MG}_{s}$ of the MGBC-CM is the
closure of the set of all achievable rate pairs $(R_1, R_2)$.

\section{Main Result} \label{sec:main}

The two-user Gaussian BC with multiple transmit-antennas is non-degraded. For
this channel, we have the following closed-from result on the secrecy capacity
region under the information-theoretic secrecy requirement.

\begin{theorem} \label{thm:GBC}
Consider an MGBC-CM modeled as in (\ref{eq:miso}). Let
\begin{align*}
\gamma_1(\alpha)=\frac{1+\alpha P|\hv^{H}\ev_1|^2}{1+ \alpha P
|\gv^{H}\ev_1|^2}, 
\end{align*}
$\gamma_2(\alpha)$ be the largest generalized eigenvalue of the pencil
\begin{align}
\left(I+\frac{(1-\alpha) P \gv\gv^{H}}{1+\alpha P |\gv^{H}\ev_1|^2},\,
I+\frac{(1-\alpha) P\hv \hv^{H}}{1+ \alpha P |\hv^{H}\ev_1|^2}\right),
\label{eq:pencil2}
\end{align}
and ${\Rc}^{\rm MG}(\alpha)$ denote the union of all $(R_1,R_2)$ satisfying
\begin{align*}
&          &  0&\le R_1\le \log_2 \gamma_1(\alpha)  &\notag\\
&\text{and} & 0&\le R_2\le \log_2 \gamma_2(\alpha).  & 
\end{align*}
Then the secrecy capacity region of the MGBC-CM is
\begin{align*}
\Cc^{\rm MG}_{s}= {\rm co} \left\{\bigcup_{0\le \alpha \le 1} \Rc^{\rm
MG}(\alpha)\right\} 
\end{align*}
where ${\rm co}\{\Sc\}$ denotes the convex hull of the set $\Sc$.
\end{theorem}
\begin{proof}
The achievability part of Theorem~\ref{thm:GBC} is based on secret dirty-paper
coding inner bound in Sec.~\ref{sec:in}. The converse part is based on
Sato-type outer bound in Sec.~\ref{sec:out}. We provide the complete proof in
\cite{Liu:IT:2007oct}.
\end{proof}

\begin{corollary} \label{cor:maxr1}
For the MGBC-CM, the maximum secrecy rate of user 1 is given by
\begin{align*}
R_{1,\max} =\max_{0\le\alpha\le1} \log_2\gamma_1(\alpha) = \log_2\lambda_1
\end{align*}
where $\lambda_1$ is defined in (\ref{eq:eig-def1}).
\end{corollary}

{\it Example:} (MISO wiretap channels) A special case of the MGBC-CM model is
the Gaussian MISO wiretap channel studied in
\cite{Li:CISS:07,Wornell:ISIT:07,Ulukus:ISIT:07}, where the transmitter sends
confidential information to only one user and treats another user as an
eavesdropper. Let us consider a Gaussian MISO wiretap channel modeled in
(\ref{eq:miso}), where user 1 is the legitimate receiver and user 2 is the
eavesdropper. Corollary~\ref{cor:maxr1} implies that the secrecy capacity of
the Gaussian MISO wiretap channel corresponds to the corner point of $\Cc^{\rm
MG}_{s}$. Hence, the secrecy capacity of the Gaussian MISO wiretap channel is
given by
\begin{align*}
C^{\rm MISO}_{s}=\log_2 \lambda_1,
\end{align*}
which coincides with the result of \cite{Wornell:ISIT:07}.

For the MGBC-CM, the actions of user 1 and user 2 are symmetric to each other,
i.e., each user decodes its own message and eavesdrops upon the confidential
information belonging to the other user. Based on symmetry of this two-user BC
model, we can express the secrecy capacity region $\Cc^{\rm MG}_{s}$ in an
alternative way.
\begin{corollary} \label{cor:GBC-2}
For an MGBC-CM modeled in as (\ref{eq:miso}), the secrecy capacity region can
be written as
\begin{align*}
\Cc^{\rm MG}_{s}= {\rm co} \left\{\bigcup_{0\le \beta \le 1} \Rc^{\rm
MG-2}(\beta)\right\} 
\end{align*}
where ${\Rc}^{\rm MG-2}(\beta)$ denotes the union of all $(R_1,R_2)$ satisfying
\begin{align*}
&          &  0&\le R_1\le \log_2 \xi_1(\beta) &\notag\\
&\text{and} & 0&\le R_2\le \log_2 \xi_2(\beta), & 
\end{align*}
$\xi_1(\beta)$ is the largest generalized eigenvalue of the pencil
\begin{align*}
\left(I+\frac{(1-\beta) P \hv\hv^{H}}{1+\beta P |\hv^{H}\ev_2|^2},\,
I+\frac{(1-\beta) P \gv \gv^{H}}{1+ \beta P |\gv^{H}\ev_2|^2}\right)
\end{align*}
and
\begin{align*}
\xi_2(\beta)=\frac{1+\beta P|\gv^{H}\ev_2|^2}{1+ \beta P |\hv^{H}\ev_2|^2}.
\end{align*}
\end{corollary}

\begin{remark}
Theorem~\ref{thm:GBC} and Corollary~\ref{cor:GBC-2} imply that if $\alpha$ and
$\beta$ satisfy the implicit function $\gamma_1(\alpha)=\xi_1(\beta)$, then
$${\Rc}^{\rm MG}(\alpha)={\Rc}^{\rm MG-2}(\beta).$$
For example, it is easy to check ${\Rc}^{\rm MG}(1)={\Rc}^{\rm MG-2}(0)$.
\end{remark}

Now, by applying Corollary~\ref{cor:GBC-2} and setting $\beta=1$, we can show
that the rate pair $(0,  \log_2\lambda_2)$ is the corner point corresponding to
the maximum achievable rate of user 2 in the capacity region $\Cc^{\rm
MG}_{s}$.
\begin{corollary} \label{cor:maxr2}
For the MGBC-CM, the maximum secrecy rate of user 2 is given by
$$R_{2,\max}=\log_2\lambda_2$$
where $\lambda_2$ is defined in (\ref{eq:eig-def2}).
\end{corollary}
Corollaries~\ref{cor:maxr1} and~\ref{cor:maxr2} imply that for the MGBC-CM,
both users can achieve positive rates with information-theoretic secrecy if and
only if $\lambda_1>1$ and $\lambda_2>1$. Furthermore, Lemma~\ref{lem:geig}
illustrates that this condition can be ensured when the attenuation vectors
$\hv$ and $\gv$ are linearly independent.

\section{Achievability: Secret DPC Scheme} \label{sec:in}

\subsection{Double-Binning Inner bound for the BC-CM}

An achievable rate region for the broadcast channel with confidential messages
(BC-CM) has been established in \cite{Liu07it} based on a double-binning scheme
that enables both joint encoding at the transmitter by using Gel'fand-Pinsker
binning and preserving confidentiality by using random binning.

\begin{lemma} (\cite[Theorem~3]{Liu07it}) \label{lem:InBC}
Let $\Vv_1$ and $\Vv_2$ be auxiliary random variables, $\Omega$ denote the
class of joint probability densities $p(\vv_1,\vv_2,\xv,y_1,y_2)$ that factor
as
\begin{align*}
p(\vv_1,\vv_2)p(\xv|\vv_1,\vv_2)p(y_1,y_2|\xv),
\end{align*}
and ${\Rc}_{\rm I}(\pi)$ denote the union of all $(R_1,R_2)$ satisfying
\begin{align*}
&            & 0 &\le R_1 \le I(\Vv_1;Y_1)-I(\Vv_1;Y_2,\Vv_2) &\\
& \text{and} & 0 &\le R_2 \le I(\Vv_2;Y_2)-I(\Vv_2;Y_1,\Vv_1) 
\end{align*}
for a given joint probability density $\pi\in \Omega$. For the BC-CM, any rate
pair
\begin{align}
(R_1,R_2)\in {\rm co} \left\{\bigcup_{\pi\in\Omega} \Rc_{\rm I}(\pi)\right\}
\label{eq:inner}
\end{align}
is achievable.
\end{lemma}

\subsection{Secret DPC Scheme for the MGBC-CM}

The achievable strategy in Lemma~\ref{lem:InBC} introduces a double-binning
coding scheme. However, when the rate region (\ref{eq:inner}) is used as a
constructive technique, it not clear how to choose the auxiliary random
variables $\Vv_1$ and $\Vv_2$ to implement the double-binning codebook, and
hence, one has to ``guess'' the density of $p(\vv_1,\vv_2,\xv)$. Here, we
employ the DPC technique with the double-binning code structure to develop the
{\it secret DPC} (S-DPC) achievable rate region.

For the MGBC-CM, we consider a secret dirty-paper encoder with Gaussian
codebooks. Based on Lemma~\ref{lem:InBC}, we obtain a S-DPC rate region for the
MGBC-CM as follows.
\begin{lemma}\label{lem:GBCin} {\rm [S-DPC region]}
Let  ${\Rc}_{\rm I}^{\rm S-DPC}(K_{\Uv_1},K_{\Uv_2})$ denote the union of all
$(R_1,R_2)$ satisfying
\begin{align}
& &0&\le R_1\le \log_2 \frac{1+\hv^{H} K_{\Uv_1} \hv}{1+\gv^{H} K_{\Uv_1}
\gv}& \label{eq:dpc-r1}\\
& \text{and} & 0&\le R_2\le \log_2 \frac{1+\gv^{H} (K_{\Uv_1}+K_{\Uv_2})
\gv}{1+\hv^{H} (K_{\Uv_1}+K_{\Uv_2}) \hv} + &\notag\\
&  & &\qquad ~~\quad \log_2 \frac{1+\hv^{H} K_{\Uv_1}\hv}{1+\gv^{H} K_{\Uv_1}
\gv}. \label{eq:dpc-r2}&
\end{align}
Then, any rate pair
\begin{align*}
(R_1,R_2)\in {\rm co} \left\{\bigcup_{\tr(K_{\Uv_1}+K_{\Uv_2})\le P} \Rc_{\rm
I}^{\rm S-DPC}(K_{\Uv_1},K_{\Uv_2})\right\} 
\end{align*}
is achievable for the MGBC-CM.
\end{lemma}
\begin{proof}
A detail proof can be found in \cite{Liu:IT:2007oct}.
\end{proof}

The S-DPC achievable rate region requires optimization of the covariance
matrices $K_{\Uv_1}$ and $K_{\Uv_2}$. In order to achieve the boundary of
$\Cc^{\rm MG}_{s}$, we choose $K_{\Uv_1}$ and $K_{\Uv_2}$ as follows:
\begin{align}
&           & K_{\Uv_1}&= \alpha P \ev_1\ev_1^{H} &\notag\\
&\text{and} & K_{\Uv_2}&= (1-\alpha) P \cv_2(\alpha)\cv_2^{H}(\alpha), \quad
0\le\alpha\le1 & \label{eq:ku12}
\end{align}
where $\ev_1$ is defined in (\ref{eq:eig-def1}) and $\cv_2(\alpha)$ is a
normalized eigenvector of the pencil (\ref{eq:pencil2}) corresponding to
$\gamma_2(\alpha)$. Next, inserting (\ref{eq:ku12}) into (\ref{eq:dpc-r1}) and
(\ref{eq:dpc-r2}), we obtain
\begin{align}
&\frac{1+\hv^{H} K_{\Uv_1} \hv}{1+\gv^{H}
K_{\Uv_1} \gv}=\gamma_1(\alpha) \qquad \qquad \text{and} \notag\\
&\frac{[1+\gv^{H} (K_{\Uv_1}+K_{\Uv_2}) \gv] [1+\hv^{H} K_{\Uv_1}\hv]}
{[1+\hv^{H} (K_{\Uv_1}+K_{\Uv_2}) \hv][1+\gv^{H} K_{\Uv_1} \gv]}
=\gamma_2(\alpha). \label{eq:sdpc-r2}
\end{align}
Now, by substituting (\ref{eq:sdpc-r2}) into Lemma~\ref{lem:GBCin}, we obtain
the desired achievable result.

\section{Converse: Sato-Type Outer Bound} \label{sec:out}

\subsection{Sato-Type Outer Bound}
We consider an important property for the BC-CM in the following lemma.
\begin{lemma} \label{lem:sam}
Let $\Pc$ denote the set of channels $p_{\Yt_1,\Yt_2|\Xv}$ whose marginal
distributions satisfy
\begin{align*}
&            & p_{\Yt_1|\Xv}(y_1|\xv)&=p_{Y_1|\Xv}(y_1|\xv) &\notag \\
&\text{and}  & p_{\Yt_2|\Xv}(y_2|\xv)&=p_{Y_1|\Xv}(y_2|\xv) &
\end{align*}
for all $y_1$, $y_2$ and $\xv$. The secrecy capacity region $\Cc^{\rm MG}_{s}$
is the same for the channels $p_{\Yt_1,\Yt_2|\Xv} \in \Pc$.
\end{lemma}

We note that $\Pc$ is the set of channels $p_{\Yt_1,\Yt_2|\Xv}$ that have the
same marginal distributions as the original channel transition density
$p_{Y_1,Y_2|\Xv}$. Lemma~\ref{lem:sam} implies that the secrecy capacity region
$\Cc^{\rm MG}_{s}$ depends only on marginal distributions.

\begin{theorem} \label{thm:out1} 
Let $\Rc_{\rm O}\bigl(P_{\Yt_1,\Yt_2|\Xv}, P_{\Xv}\bigr)$ denote the union of
all rate pairs $(R_1, R_2)$ satisfying
\begin{align*}
&     &        R_1&\le I(\Xv;\Yt_1|\Yt_2) & \\
& \text{and} & R_2&\le I(\Xv;\Yt_2|\Yt_1) & 
\end{align*}
for given distributions $P_{\Xv}$ and $P_{\Yt_1,\Yt_2|\Xv}$. The secrecy
capacity region $\Cc^{\rm MG}_{s}$ of the BC-CM satisfies
\begin{align}
\Cc^{\rm MG}_{s} \subseteq \bigcap_{P_{\Yt_1,\Yt_2|\Xv}\in \Pc}
\left\{\bigcup_{P_{\Xv}} \Rc_{\rm O}\bigl(P_{\Yt_1,\Yt_2|\Xv},
P_{\Xv}\bigr)\right\}. \label{eq:Sato}
\end{align}
\end{theorem}
\begin{proof}
A detail proof can be found in \cite{Liu:IT:2007oct}.
\end{proof}
\begin{remark}
The outer bound (\ref{eq:Sato}) follows by evaluating the secrecy level at each
receiver end in an individual manner, while letting the users decode their
messages in a \emph{cooperative} manner. In this sense, we refer to this bound
as ``Sato-type'' outer bound.
\end{remark}
For example, we consider the confidential message $W_1$ that is destined for
user 1 (corresponding to $\Yt_1$) and eavesdropped upon by user 2
(corresponding to $\Yt_2$). We assume that a genie gives user 1 the signal
$\Yt_2$ as side information for decoding $W_1$. Note that the eavesdropped upon
signal $\Yt_2$ at user 2 is always a degraded version of the entire received
signal $(\Yt_1,\Yt_2)$. This permits the use of the wiretap channel result of
\cite{Wyner:BSTJ:75}.

\begin{remark}
Although Theorem~\ref{thm:out1} is based on a \emph{degraded} argument, the
outer bound (\ref{eq:Sato}) can be applied to the \emph{general} broadcast
channel with confidential messages.
\end{remark}

\subsection{Sato-Type Outer Bound for the MGBC-CM}

For the Gaussian BC, the family  $\Pc$ is the set of channels
\begin{align*}
\yt_1&=\hv^{H} \xv+\zt_1 \notag\\
\yt_2&=\gv^{H} \xv+\zt_2 
\end{align*}
where $\zt_1$ and $\zt_2$ correspond to arbitrarily correlated, zero-mean,
unit-variance, complex Gaussian random variables. Let $\rho$ denote the
covariance between $\Zt_1$ and $\Zt_2$, i.e,
$$\Cov\bigl(\Zt_1,\Zt_2\bigr)=\rho \quad \text{and} \quad |\rho|^2\le1.$$
Now, the rate region ${\Rc}_{\rm O}\bigl(P_{\Yt_1,\Yt_2|\Xv}, P_{\Xv}\bigr)$ is
a function of the noise covariance $\rho$ and the input covariance matrix
$K_{\Xv}$. We consider a computable Sato-type outer bound for the MGBC-CM in
the following lemma.
\begin{lemma} \label{lem:outG}
Let ${\Rc}_{\rm O}^{\rm MG}(\rho, K_{\Xv})$ denote the union of all rate pairs
$(R_1,R_2)$ satisfying
\begin{align*}
&            & 0 \le R_1&\le f_1(\rho,K_{\Xv}) &\\
& \text{and} & 0 \le R_2&\le f_2(\rho,K_{\Xv})&
\end{align*}
where
\begin{align*}
f_1(\rho,K_{\Xv}) = \min_{\nu \in \mathbb{C}} & \log_2
\frac{(\hv-\nu\gv)^{H} K_{\Xv} (\hv-\nu\gv)+\psi_1(\nu,\rho)}{(1-|\rho|^2)}\notag\\
 f_2(\rho,K_{\Xv})= \min_{\mu \in \mathbb{C}} & \log_2 \frac{(\gv-\mu\hv)^{H}
K_{\Xv} (\gv-\mu\hv)+\psi_2(\mu,\rho)}{(1-|\rho|^2)} \notag\\
\psi_1(\nu,\rho)&=1+|\nu|^2-\nu^{*}\rho-\rho^{*}\nu \notag\\
\text{and} \quad \psi_2(\mu,\rho)&=1+|\mu|^2-\mu^{*}\rho-\rho^{*}\mu.
\end{align*}
For the MGBC-CM, the secrecy capacity region $\Cc^{\rm MG}_{s}$ satisfies
\begin{align*}
\Cc^{\rm MG}_{s} \subseteq \bigcup_{\tr(K_{\Xv})\le P} {\Rc}_{\rm
O}(\rho,K_{\Xv}) 
\end{align*}
for any $0\le|\rho|\le1$.
\end{lemma}
\begin{remark}
Lemma~\ref{lem:outG} is based on the fact that Gaussian input distributions
maximize ${\Rc}_{\rm O}\bigl(P_{\Yt_1,\Yt_2|\Xv}, P_{\Xv}\bigr)$ for Gaussian
broadcast channel. To illustrate this point, we consider
\begin{align*}
I(\Xv;\Yt_1|\Yt_2)
&=h(\Yt_1|\Yt_2)-\log_2 (2\pi e)(1-|\rho|^2)\notag\\
&\le h(\Yt_1-\nu \Yt_2)-\log_2 (2\pi e)(1-|\rho|^2). 
\end{align*}
Moreover, the maximum-entropy theorem \cite{Cover} implies that $h(\Yt_1-\nu
\Yt_2)$ is maximized by Gaussian input distributions.
\end{remark}

Finally, we prove that the Sato-type outer bound of Lemma~\ref{lem:outG}
coincides with the secrecy capacity region $\Cc^{\rm MG}_{s}$ by choosing the
parameter $\rho=(\gv^{H}\ev_1)/(\hv^{H}\ev_1).$ A detail proof can be found in
\cite{Liu:IT:2007oct}.

\section{Numerical Examples} \label{sec:ex}

In this section, we study a numerical example to illustrate the secrecy
capacity region of the MGBC-CM. For simplicity, we assume that the Gaussian BC
has real input and output alphabets and the channel attenuation vectors $\hv$
and $\gv$ are also real. Under these conditions, all calculated rate values are
divided by $2$.

\begin{figure}[hbt]
 \centerline{\includegraphics[width=0.9\linewidth,draft=false]{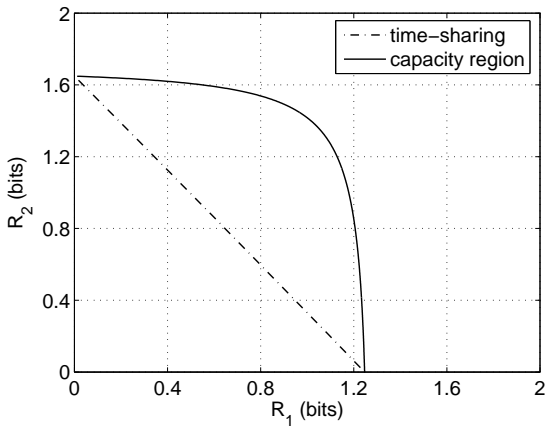}}
\caption{Secrecy capacity region vs. time-sharing secrecy rate region
 for the example MGBC-CM in (\ref{eq:ex1})}
  \label{fig:sim1}
\end{figure}
In particular, we consider the following MGBC-CM
\begin{align}
\left[\begin{matrix}y_1 \\ y_2 \end{matrix}\right] &= \left[\begin{matrix} 1.5
&0 \\ 1.801 & 0.871 \end{matrix}\right] \left[\begin{matrix}x_1 \\ x_2
\end{matrix}\right] +\left[\begin{matrix}z_1 \\ z_2
\end{matrix}\right]\label{eq:ex1}
\end{align}
where $\hv=[1.5, 0]^{T}$, $\gv=[1.801, 0.872]^{T}$, and the total power
constraint is set to $P=10$. Fig.~\ref{fig:sim1} illustrates the secrecy
capacity region for the channel (\ref{eq:ex1}). We observe that even though
each component of the attenuation vector $\hv$ (imposed on user 1) is strictly
less than the corresponding component of $\gv$ (imposed on user 2), both users
can achieve positive rates simultaneously under the information-theoretic
secrecy requirement. Moreover, we compare the secrecy capacity region with the
secrecy rate region achieved by the time-sharing scheme (indicated by the
dash-dot line). Fig.~\ref{fig:sim1} demonstrate that the time-sharing scheme is
strictly suboptimal for providing the secrecy capacity region.

\end{document}